\documentclass[journal,subeqn]{IEEEtran}
\usepackage{cite}
\usepackage{amsfonts,amssymb,amsmath,amsthm}
\usepackage[bb=boondox]{mathalfa}
\usepackage{calrsfs}
\usepackage{latexsym}
\usepackage{epsfig}
\usepackage{epstopdf}
\usepackage[font=footnotesize]{subfig}
\usepackage{cite,color,framed}
\usepackage{graphicx}
\usepackage{subeqn}
\usepackage{algorithmicx, algpseudocode}
\usepackage{float} 

\newtheorem{prob-statement}{Problem}

\newtheorem{lemma}{Lemma}
\newtheorem{thrm}{Theorem}

\newtheorem{prop}{Proposition}
\DeclareMathOperator*{\argmax}{arg\,max}
\DeclareMathOperator*{\argmin}{arg\,min}

\DeclareMathOperator*{\tr}{Tr}

\DeclareMathOperator*{\sign}{sign}

\DeclareMathAlphabet{\mathpzc}{OT1}{pzc}{m}{it}

\begin{document}

\title{On Strategic Multi-Antenna Jamming in 
\\Centralized Detection Networks}
%

\author{
\IEEEauthorblockN{V. Sriram Siddhardh Nadendla, \IEEEmembership{Student~Member,~IEEE}, Vinod Sharma, \IEEEmembership{Senior~Member,~IEEE}
\\
and Pramod K. Varshney, \IEEEmembership{Fellow,~IEEE}
}

\thanks{V. Sriram Siddhardh Nadendla and Pramod K. Varshney are with the Department
of Electrical Engineering and Computer Science, Syracuse University, Syracuse, NY 13201, USA. E-mail: \{vnadendl, varshney\}@syr.edu.}
\thanks{Vinod Sharma is with the Department
of Electrical Communication Engineering, Indian Institute of Science, Bangalore 560012, India. E-mail: vinod@ece.iisc.ernet.in.}
}


\maketitle

\begin{abstract}
In this paper, we model a complete-information zero-sum game between a centralized detection network with a multiple access channel (MAC) between the sensors and the fusion center (FC), and a jammer with multiple transmitting antennas. We choose error probability at the FC as the performance metric, and investigate pure strategy equilibria for this game, and show that the jammer has no impact on the FC's error probability by employing pure strategies at the Nash equilibrium. Furthermore, we also show that the jammer has an impact on the expected utility if it employs mixed strategies.
\end{abstract}


\IEEEpeerreviewmaketitle

\section{Introduction and System Model}
Jamming attacks in detection networks have a significant impact on today's world due to the wide range of applications of these networks \cite{Veeravalli2011}. Therefore, several attempts have been made in the past literature to address jamming attacks in detection networks. For more details, please refer to \cite{Perrig2002,Perrig2004,Mpitziopoulos2009} and citations within. In our past work, we have addressed jamming attacks in the context of detection networks with multiple access channels (MACs) \cite{Nadendla2010,Nadendla2011}. In particular, we found equilibrium strategies numerically for a zero-sum game between a centralized detection network and a simple Gaussian jammer with an average power constraint and a single antenna per channel in \cite{Nadendla2011}. In this paper, we extend our work in \cite{Nadendla2011} by investigating pure strategy equilibria in closed form, for a complete-information zero-sum game between a \emph{centralized} detection network and a powerful jammer equipped with multiple antennas and strict (instantaneous) power constraints.

Consider a \emph{centralized} detection network where $N$ sensing agents share raw observations with the fusion center (FC) which makes a global decision regarding the presence/absence of the phenomenon-of-interest (PoI) in the presence of a disruptive jammer, as shown in Figure \ref{Fig: model}. Let $H_1$ denote the hypothesis when PoI is present, and $H_0$ otherwise, with prior probabilities $\pi_1$ and $\pi_0$ respectively. We model the PoI's signal as $\theta = 1$ under $H_1$, and $\theta = 0$ otherwise. In this paper, we refer to the channel between the PoI and any given sensor as a \emph{sensing channel}, and the channel between the sensors and the FC as a \emph{communication channel}. We assume a multiple access channel (MAC) at the communication channel, where all the sensors' messages are superimposed into one received signal at the FC. 

The disruptive jammer interferes with both the sensing and the communication channels by introducing the jamming symbols $\boldsymbol{w}_s$ and $\boldsymbol{w}_{fc}$ respectively. For the sake of notational convenience, we stack these jamming symbols together into a super-symbol $\boldsymbol{w} = \{\boldsymbol{w}_s, \boldsymbol{w}_{fc}\}$. We assume that the jammer has a total power budget $P$, and denote the set of all possible jammer's strategies as $\mathcal{W} \triangleq \{ \boldsymbol{w} \in \mathbb{R}^{L+M} \ | \ ||\boldsymbol{w} ||_2^2 \leq P \}$.

If $\alpha_i$ and $\beta_{il}$ denote the known channel-gains at the $i^{th}$ sensing channel due to the PoI signal and the $l^{th}$ antenna at the jammer respectively, the $i^{th}$ sensor acquires an observation
\begin{equation}
	s_i = \displaystyle \alpha_i \theta + \sum_{l = 1}^L \beta_{il} w_{s_l} + n_i,
	\label{Eqn: Sensor-Observation}
\end{equation}
where $n_i$ is a zero-mean AWGN noise with variance $\sigma_s^2$. 

\begin{figure}[!t]
	\centering
    \includegraphics[width=3.3in]{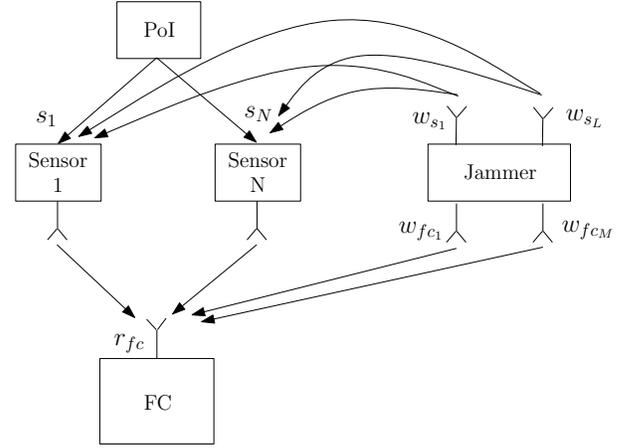}
    \caption{Detection Network in the Presence of a Jammer}
    \label{Fig: model}
\end{figure}

We assume that the $i^{th}$ sensor transmits its raw observation $s_i$ over the MAC. The FC receives the combined signal
\begin{equation}
	\begin{array}{lcl}
    	r_{fc} & = & \displaystyle \sum_{i = 1}^{N} \phi_i s_i + \sum_{m = 1}^{M} \psi_m w_{fc_m} + n_{fc}
    	\\[3ex]
    	& = & \displaystyle a \theta + \boldsymbol{b}^T \boldsymbol{w} + z,
    \end{array}
\end{equation}
\begin{equation*}
	\begin{array}{l}
		\mbox{where} \quad a = \displaystyle \sum_{i = 1}^{N} \phi_i \alpha_i, \quad z = \displaystyle \sum_{i = 1}^{N} \phi_i n_i + n_{fc}, \quad \mbox{and}
		\\[3ex]
		\boldsymbol{b}^T = \displaystyle \left[ 
			\begin{array}{cccccc}
				\displaystyle \sum_{i = 1}^{N} \phi_i \beta_{i1}
				& 
				\cdots
				& 
				\displaystyle \sum_{i = 1}^{N} \phi_i \beta_{iL}
				& 
				\psi_1
				& 
				\cdots
				& 
				\psi_M
			\end{array}
		\right].
	\end{array}
	\label{Eqn: a,b,z}
\end{equation*}

Since $r_{fc}$ is a superposition of the PoI's signal with several Gaussian random variables, $r_{fc} | H_0 \sim \mathcal{N}(\boldsymbol{b}^T\boldsymbol{w}, \sigma^2)$ and $\boldsymbol{r}_{fc} | H_1 \sim \mathcal{N}(a + \boldsymbol{b}^T\boldsymbol{w}, \sigma^2)$, where $\sigma^2 = \displaystyle \sigma_{fc}^2 + \sigma_s^2 \sum_{i = 1}^{N} \phi_i^2$ is the variance of the noise signal $z$. 

We assume that the FC employs a decision rule\footnote{Since this is a likelihood ratio test, all the other rules are dominated. Therefore, their removal does not result any loss in network performance.}
\begin{equation}
	r_{fc} \  \mathop{\stackrel{H_1}{\gtrless}}_{H_0} \ \lambda,
	\label{Eqn: FC-Test-Statistic}
\end{equation}
where $\lambda \in \Lambda$\footnote{Although $\lambda$ can be any real number in practice, for the sake of tractability, we assume that $\Lambda \triangleq [-R, R]$, where $R$ is a sufficiently large real number. For more details, the reader may refer to Theorem 5, Page 168 in \cite{Book-Basar} which guarantees the existence of a mixed strategy equilibrium.} is a real-valued threshold designed to minimize the FC's error probability 
\begin{equation}
	P_E = \displaystyle \pi_0 Q \left( \frac{\lambda - \boldsymbol{b}^T \boldsymbol{w}}{\sigma} \right) + \pi_1 \left[ 1 - Q \left( \frac{\lambda - \boldsymbol{b}^T \boldsymbol{w} - a}{\sigma} \right) \right],
	\label{Eqn: Pe-FC-centralized-K=1}
\end{equation}
while the jammer simultaneously attempts to maximize $P_E$ by employing an appropriate jamming signal $\boldsymbol{w}$. 


\section{Evaluation of Pure Strategy Equilibria}
We model the interaction between the FC and the jammer formally as a zero-sum game as stated below.
\begin{prob-statement}
Find the Nash equilibria $\{ \lambda^*, \boldsymbol{w}^* \} \in \Lambda \times \mathcal{W}$ that satisfy the following inequality:
\begin{align*}
	\begin{array}{c}
		\\[-1.5ex]
		\displaystyle P_E(\lambda^*,\boldsymbol{w}) \leq P_E(\lambda^*,\boldsymbol{w}^*) \leq P_E(\lambda,\boldsymbol{w}^*)
		\\[2ex]
		\forall \ \lambda \in \Lambda, \ \boldsymbol{w} \in \mathcal{W}.
	\end{array}
\end{align*}
\label{Prob.Stmt: K=1}
\end{prob-statement}

First, we investigate some important properties of $P_E$ that guarantee the existence of pure-strategy equilibria.

\begin{lemma}
For a given $\boldsymbol{b}$, $\boldsymbol{w}$ and $\sigma$, $P_E$ is a quasiconvex function of $\lambda$.
\label{Lemma: Quasiconvex P_E(lambda) when K=1} 
\end{lemma}

\begin{proof}
For a fixed $\boldsymbol{b}$, $\boldsymbol{w}$ and $\sigma$, we differentiate $P_E$ with respect to $\lambda$ and obtain
\begin{equation}
	\begin{array}{lcl}
		\displaystyle \frac{\partial P_E}{\partial \lambda} & = & \displaystyle 
		f_1(\lambda) \cdot \left[ \pi_1 f_2(\lambda) - \pi_0 \right]
	\end{array}
	\label{Eqn: diff-Pe-lambda}
\end{equation}
where
\begin{subequations}
\begin{equation}
	\label{Eqn: f1_centralized}
	f_1(\lambda) = \displaystyle \frac{1}{\sigma\sqrt{2\pi}} \exp{\left(-\frac{(\lambda - \boldsymbol{b}^T \boldsymbol{w})^2}{2\sigma^2}\right)},
\end{equation}
\begin{equation}
	\label{Eqn: f2_centralized}
	f_2(\lambda) = \exp{\left( \frac{\displaystyle 2 a (\lambda - \boldsymbol{b}^T \boldsymbol{w}) - a^2}{2 \sigma^2} \right)}.
\end{equation}
\end{subequations}

Since this structure has similar properties as that in Lemma 1 in \cite{Zhang2002}, we omit the remaining proof for brevity.
%
\end{proof}

Any channel model with non-negative channel gains ensures that every element in the vector $\boldsymbol{b}$ is non-negative. Since many practical channel models such as path-loss model and Rayleigh fading model have non-negative channel gains, we assume that $\boldsymbol{b}$ is a non-negative vector in the rest of this section. 

\begin{lemma}
For a given $\lambda$, $\boldsymbol{b}$ and $\sigma$, $P_E$ is jointly quasiconcave in $\boldsymbol{w}$, if every entry in $\boldsymbol{b}$ is non-negative.
\label{Lemma: Quasiconcave P_E(w) when K=1} 
\end{lemma}

\begin{proof}
Given any two points $\boldsymbol{w}_1, \boldsymbol{w}_2 \in \mathcal{W}$, $P_E$ is jointly quasiconcave \cite{Book-Boyd2004} if and only if
\begin{equation}
	\displaystyle P_E(\boldsymbol{w}_1) \leq P_E(\boldsymbol{w}_2) \quad \Rightarrow \quad \nabla_{\boldsymbol{w}} P_E(\boldsymbol{w}_1) \cdot (\boldsymbol{w}_1 - \boldsymbol{w}_2) \leq 0.
	\label{Eqn: Def-Quasiconcavity}
\end{equation}

In our framework, the necessary condition $\displaystyle P_E(\boldsymbol{w}_1) \leq P_E(\boldsymbol{w}_2)$ reduces to
\begin{equation}
	\displaystyle \int_{y_2}^{y_1} g(y) dy \geq 0.
	\label{Eqn: Necessary-Condition-Quasiconcave-Pe-K=1}
\end{equation}
Here, $y_1 = \boldsymbol{b}^T \boldsymbol{w}_1$ and $y_2 = \boldsymbol{b}^T \boldsymbol{w}_2$ are the integral limits 
and
\begin{equation}
	g(y) = f_3(y) \cdot \left[ \pi_1 f_4(y) - \pi_0 \right],
\end{equation}
where
\begin{subequations}
\begin{equation}
	\label{Eqn: f3_centralized}
	f_3(y) = \displaystyle \frac{1}{\sigma\sqrt{2\pi}} \exp{\left(-\frac{(\lambda - y)^2}{2\sigma^2}\right)},
\end{equation}
\begin{equation}
	\label{Eqn: f4_centralized}
	f_4(y) = \exp{\left( \frac{\displaystyle 2 a (\lambda - y) - a^2}{2 \sigma^2} \right)}.
\end{equation}
\end{subequations}
%

Given that the values of $\boldsymbol{b}$, $\lambda$ and $\sigma$ are fixed, we differentiate $P_E$ with respect to $\boldsymbol{w}$ and obtain
\begin{equation}
	\begin{array}{lcl}
		\displaystyle \nabla_{\boldsymbol{w}} P_E(\boldsymbol{w}_1) & = & \displaystyle 
		-\boldsymbol{b} \cdot g(y_1).
	\end{array}
	\label{Eqn: diff-Pe-w-K=1}
\end{equation}

Substituting Equations \eqref{Eqn: Necessary-Condition-Quasiconcave-Pe-K=1} and \eqref{Eqn: diff-Pe-w-K=1} in Equation \eqref{Eqn: Def-Quasiconcavity}, we need to show that
\begin{equation}
	\displaystyle \displaystyle \int_{y_2}^{y_1} g(y) dy \geq 0 \quad \Rightarrow \quad g(y_1) \cdot \left[ y_1 - y_2 \right] \ \geq \ 0
	\label{Eqn: Prove-Quasiconcave-Pe-K=1}
\end{equation}
in order to prove the lemma.

%

%

Note that $f_3(y) \geq 0$. Since $f_4(y)$ is a monotonically decreasing function of $y$, we have $g(y) \geq 0$ whenever $y \leq y_0$, and $g(y) < 0$ whenever $y > y_0$, where $y_0$ is the unique zero-crossing point at which $f_4(y_0) = \displaystyle \frac{\pi_0}{\pi_1}$. Using this property, we prove the theorem in three cases, as shown below. 

\paragraph*{CASE-1 [$y_0 \leq y_1, y_2$]} Given that $y_0 \leq y_1, y_2$, we have $g(y) \leq 0$ for any $y$ between $y_1$ and $y_2$. In such a case, the necessary condition given in Equation \eqref{Eqn: Necessary-Condition-Quasiconcave-Pe-K=1} holds true when $y_1 \leq y_2$. In other words, $g(y_1) \cdot [y_1 - y_2] \geq 0$ whenever Equation \eqref{Eqn: Necessary-Condition-Quasiconcave-Pe-K=1} holds true in this case. 

\paragraph*{CASE-2 [$y_1, y_2 \leq y_0$]} Given that $y_1, y_2 \leq y_0$, we have $g(y) \geq 0$ for any $y$ between $y_1$ and $y_2$. Therefore, the necessary condition in Equation \eqref{Eqn: Necessary-Condition-Quasiconcave-Pe-K=1} holds true when $y_2 \leq y_1$. As a result, $g(y_1) \cdot [y_1 - y_2] \geq 0$ whenever Equation \eqref{Eqn: Necessary-Condition-Quasiconcave-Pe-K=1} holds true in this case.

\paragraph*{CASE-3 [$y_1 \leq y_0 \leq y_2$ or $y_2 \leq y_0 \leq y_1$]} Note that this is a trivial case. This is because of the following. If $y_1 \leq y_0 \leq y_2$, both $g(y_1)$ and $(y_1 - y_2)$ are negative. On the other hand, if $y_2 \leq y_0 \leq y_1$, both $g(y_1)$ and $(y_1 - y_2)$ are positive. Either way, their product $g(y_1) \cdot [y_1 - y_2] \geq 0$ whether or not, the necessary condition in Equation \eqref{Eqn: Necessary-Condition-Quasiconcave-Pe-K=1} holds true.
\end{proof}

Given that $P_E$ is quasi-concave-convex in nature, a pure strategy solution exists due to the classic Debreu-Glicksberg-Fan existence theorem \cite{Book-Fudenberg, Book-Basar}. Therefore, we start by investigating the best response strategies at the network in the following proposition.

\begin{prop}
The optimal threshold $\lambda^* = \displaystyle \argmin_{\lambda} P_E(\lambda, \boldsymbol{w})$ for a fixed jammer's strategy $\boldsymbol{w}$ is given by
\begin{equation}
	\lambda^* = \boldsymbol{b}^T \boldsymbol{w} + c
	\label{Eqn: optimal-lambda-K=1}
\end{equation}
where $c = \displaystyle \frac{1}{2a} \left[a^2 + 2 \sigma^2 \log \left( \frac{\pi_0}{\pi_1} \right) \right]$ is a constant. Furthermore, $P_E(\lambda = \lambda^*,\boldsymbol{w})$ is independent of $\boldsymbol{w}$.
\label{Prop: Optimal-lambda-K=1}
\end{prop}
\begin{proof}
We first consider the inner optimization in the max-min problem where we minimize $P_E$ with respect to $\lambda$ for a fixed jammer's strategy $\boldsymbol{w}$. The optimal $\lambda = \lambda^*$ satisfies
\begin{equation}
	\begin{array}{lcl}
		\displaystyle \frac{\partial P_E}{\partial \lambda} & = & f_1(\lambda) \cdot \left[ \pi_1 f_2(\lambda) - \pi_0 \right] = 0,
	\end{array}
	\label{Eqn: diff-Pe-lambda}
\end{equation}
where $f_1(\lambda) \geq 0$. Thus, if $\displaystyle f_2(\lambda) = \frac{\pi_0}{\pi_1}$, we have $\displaystyle \frac{\partial P_E}{\partial \lambda} = 0$. Substituting Equation \eqref{Eqn: f2_centralized} and rearranging terms, we have
\begin{equation}
	\lambda^* = \boldsymbol{b}^T \boldsymbol{w} + c
	\label{Eqn: optimal-lambda-K=1}
\end{equation}
where $c = \displaystyle \frac{1}{2a} \left[a^2 + 2 \sigma^2 \log \left( \frac{\pi_0}{\pi_1} \right) \right]$ is independent of $\boldsymbol{w}$.

Given a fixed jammer's strategy $\boldsymbol{w}$, if the FC employs the optimal threshold $\lambda^*$, from Equation \eqref{Eqn: optimal-lambda-K=1}, the error probability at the FC is given by
\begin{equation}
	\begin{array}{lcl}
		P_E(\lambda^*,\boldsymbol{w}) & = & \displaystyle \pi_0 Q \left( \frac{c}{\sigma} \right) + \pi_1 \left[ 1 - Q \left( \frac{c - a}{\sigma} \right) \right].
	\end{array}
	\label{Eqn: Pe-FC-centralized-K=1}
\end{equation}
Note that $P_E(\lambda^*,\boldsymbol{w})$ is independent of the jammer's strategy $\boldsymbol{w}$, as stated in the proposition statement. 
\end{proof}

Note that the best response strategy employed by the network, as shown in Equation \eqref{Eqn: optimal-lambda-K=1}, is unique for a fixed jammer's strategy $\boldsymbol{w}$. Furthermore, the jammer's signal introduces a linear shift to the point $\lambda = c$, which is optimal in the absence of the jammer. In contrast, when we investigate the optimal jammer's strategy $\boldsymbol{w}^*$ by considering the min-max framework, we have the following proposition.
\begin{prop}
The optimal jammer's strategy $\boldsymbol{w}^* = \displaystyle \argmax_{\boldsymbol{w}} P_E(\lambda, \boldsymbol{w})$ for a fixed threshold $\lambda$ satisfies
\begin{equation}
	\boldsymbol{b}^T \boldsymbol{w}^* = \lambda - c.
	\label{Eqn: optimal-w-K=1}
\end{equation}
where $c = \displaystyle \frac{1}{2a} \left[a^2 + 2 \sigma^2 \log \left( \frac{\pi_0}{\pi_1} \right) \right]$. Such a pure-strategy solution exists only when 
\begin{equation}
	c - \sqrt{P \cdot \boldsymbol{b}^T \boldsymbol{b} } \ \leq \ \lambda \ \leq \ c + \sqrt{P \cdot \boldsymbol{b}^T \boldsymbol{b} }.
	\label{Eqn: existence-optimal-w-K=1}
\end{equation}
\label{Prop: Optimal-w-K=1}
\end{prop}
\begin{proof}
An approach similar to the proof of Proposition \ref{Prop: Optimal-lambda-K=1} can be followed for finding Equation \eqref{Eqn: optimal-w-K=1}. Therefore, we focus our attention on finding the existence condition, given in Equation \eqref{Eqn: existence-optimal-w-K=1}.

In order for a pure-strategy solution to exist, $\boldsymbol{w}^*$ should lie within the set of strategies that satisfy the jammer's total power budget. In other words, we need $\left( \boldsymbol{w}^* \right)^T \boldsymbol{w}^* \leq P$. Therefore, the affine function given in Equation \eqref{Eqn: optimal-w-K=1} should be within the squared-distance of $P$ units from the origin $\boldsymbol{w} = \boldsymbol{0}$. In other words, we have
\begin{equation}
	\frac{(\lambda - c)^2}{\boldsymbol{b}^T \boldsymbol{b}} \leq P.
	\label{Eqn: existence-optimal-w-K=1-equiv}
\end{equation}
Note that this condition can also be equivalently stated as given in Equation \eqref{Eqn: existence-optimal-w-K=1}.
\end{proof}

Note that the jammer's best response strategy is not unique, as shown in Equation \eqref{Eqn: optimal-w-K=1}. Indeed, there are infinite possibilities since the jammer can adopt any strategy on a line segment without any regret. Combining the results from Propositions \ref{Prop: Optimal-lambda-K=1} and \ref{Prop: Optimal-w-K=1}, we have the following main result of this section.

\begin{thrm}
For every $-\boldsymbol{b} \leq \boldsymbol{\epsilon} \leq \boldsymbol{b}$, 
\begin{equation}
	\begin{array}{lcl}
		\displaystyle \lambda^* = c + \sqrt{\frac{P}{\boldsymbol{b}^T \boldsymbol{b}}} \boldsymbol{b}^T \boldsymbol{\epsilon}, & \ & \displaystyle \boldsymbol{w}^* = \sqrt{\frac{P}{\boldsymbol{b}^T \boldsymbol{b}}} \boldsymbol{\epsilon}
	\end{array}
\end{equation}
is a pure-strategy Nash equilibrium. At the above equilibrium point, the error probability at the FC is given by
\begin{equation}
	\begin{array}{lcl}
		P_E(\lambda^*,\boldsymbol{w}^*) & = & \displaystyle \pi_0 Q \left( \frac{c}{\sigma} \right) + \pi_1 \left[ 1 - Q \left( \frac{c - a}{\sigma} \right) \right].
	\end{array}
	\label{Eqn: Pe-FC-equilibrium-K=1}
\end{equation}
\label{Thrm: NE-K=1}
\end{thrm}
\begin{proof}
As stated in Proposition \ref{Prop: Optimal-w-K=1}, $\lambda^*$ varies between $c - \sqrt{P \cdot \boldsymbol{b}^T \boldsymbol{b} }$ and $c + \sqrt{P \cdot \boldsymbol{b}^T \boldsymbol{b} }$. Therefore, we first investigate the extreme points $\lambda_1^* = c - \sqrt{P \cdot \boldsymbol{b}^T \boldsymbol{b} }$ and $\lambda_2^* = c + \sqrt{P \cdot \boldsymbol{b}^T \boldsymbol{b} }$. 

We first consider the case where $\lambda_1^* = c - \sqrt{P \cdot \boldsymbol{b}^T \boldsymbol{b} }$. Comparing this threshold to the optimal threshold from Equation \eqref{Eqn: optimal-lambda-K=1}, we have $\lambda_1^* = \boldsymbol{b}^T \boldsymbol{w} + c = c - \sqrt{P \cdot \boldsymbol{b}^T \boldsymbol{b} }$. On simplification, we find that $\boldsymbol{w}_1^* = - \sqrt{\frac{P}{\boldsymbol{b}^T \boldsymbol{b}}} \boldsymbol{b}$ is the optimal jammer's strategy. Thus, $\lambda_1^* = \boldsymbol{b}^T \boldsymbol{w} + c = c - \sqrt{P \cdot \boldsymbol{b}^T \boldsymbol{b} }$ and $\boldsymbol{w}_1^* = - \sqrt{\frac{P}{\boldsymbol{b}^T \boldsymbol{b}}} \boldsymbol{b}$ form a pure-strategy equilibrium. Similarly, it is easy to show that $\lambda_2^* = c + \sqrt{P \cdot \boldsymbol{b}^T \boldsymbol{b} }$ and $\boldsymbol{w}_2^* = \sqrt{\frac{P}{\boldsymbol{b}^T \boldsymbol{b}}} \boldsymbol{b}$ is another pure-strategy equilibrium.

Given these two pure-strategy equilibria, we find a parametric representation of all possible pure-strategy Nash equilibria, as given below. Let
\begin{equation}
	\boldsymbol{w}_{\boldsymbol{\epsilon}}^* = \sqrt{\frac{P}{\boldsymbol{b}^T \boldsymbol{b}}} \boldsymbol{\epsilon}
	\label{Eqn: optimal-w-parameterized-K=1}
\end{equation}
where $\boldsymbol{\epsilon}$ is the vector parameter that ranges from $-\boldsymbol{b}$ and $\boldsymbol{b}$. Note that the two solutions $\boldsymbol{w}_1^*$ and $\boldsymbol{w}_2^*$ both correspond to the parameter values $\boldsymbol{\epsilon}_1 = -\boldsymbol{b}$ and $\boldsymbol{\epsilon} = \boldsymbol{b}$ respectively. Furthermore, such a linear parameterization is valid because of the fact that $\boldsymbol{w}^*$ always lies on the line $\boldsymbol{b}^T \boldsymbol{w}^* = \lambda - c$, as given in Equation \eqref{Eqn: optimal-w-K=1}.

Substituting Equation \eqref{Eqn: optimal-w-parameterized-K=1} in Equation \eqref{Eqn: optimal-lambda-K=1}, we have 
\begin{equation}
	\lambda_{\boldsymbol{\epsilon}}^* = c + \sqrt{\frac{P}{\boldsymbol{b}^T \boldsymbol{b}}} \boldsymbol{b}^T \boldsymbol{\epsilon}.
\end{equation}

Since the equilibrium point satisfies the necessary conditions presented in Propositions \ref{Prop: Optimal-lambda-K=1} and \ref{Prop: Optimal-w-K=1}, the error probability at the FC is given by Equation \eqref{Eqn: Pe-FC-centralized-K=1}.
\end{proof}

\section{Discussion}

Since the network and the jammer are non-cooperative entities, we investigate the convergence of the players' strategies in a repeated game setting from any arbitrary strategy profile. We denote the initial pure strategy profile as $(\lambda_0, \boldsymbol{w}_0)$, where the total power of the initial jammer's strategy $\boldsymbol{w}_0$ is within the jammer's power budget $P$, and assume that all the players' strategies are \emph{perfectly observable}, i.e., the network makes noiseless observations regarding the jammer's strategy and vice-versa. 

\begin{lemma}
Given any initial strategy profile $(\lambda_0, \boldsymbol{w}_0)$, the players always converge to an equilibria presented in Theorem \ref{Thrm: NE-K=1} in a perfectly-observable repeated-game, irrespective of the order of their play.
\end{lemma}
\begin{proof}
We prove this lemma in two cases. In the first case, we assume that the network takes the lead, followed by the jammer and so on. In the latter case, we assume the opposite where the jammer takes the lead, followed by the network and so on.

\paragraph*{CASE-1 [N-J-N-J-$\cdots$]} In this case, we assume that the network takes the lead. Therefore, given the initial strategy profile $(\lambda_0, \boldsymbol{w}_0)$, the network chooses its best response from Proposition \ref{Prop: Optimal-lambda-K=1}, which is
\begin{equation}
	\lambda_1 = \boldsymbol{b}^T \boldsymbol{w}_0 + c.
\end{equation}
Given that $||\boldsymbol{w}_0||_2^2 \leq P$, without any loss of generality, we can represent $\boldsymbol{w}_0$ in the same form as shown in Theorem \ref{Thrm: NE-K=1}. As a result, $\lambda_1$ also has the form presented in Theorem \ref{Thrm: NE-K=1}. Thus, the repeated game converges to an equilibrium point $(\lambda_1, \boldsymbol{w}_0)$ within one iteration.

\paragraph*{CASE-2 [J-N-J-N-$\cdots$]} In this case, we assume that the jammer takes the lead. Therefore, given the initial strategy profile $(\lambda_0, \boldsymbol{w}_0)$, the jammer chooses its best response as stated in Proposition \ref{Prop: Optimal-w-K=1}. In other words, if $\lambda_0$ lies between $c - \sqrt{P \cdot \boldsymbol{b}^T \boldsymbol{b}}$ and $c + \sqrt{P \cdot \boldsymbol{b}^T \boldsymbol{b}}$, the jammer chooses its best response $\boldsymbol{w}_{1a}$ such that
\begin{equation}
	\boldsymbol{b}^T \boldsymbol{w}_{1a} = \lambda_0 - c.
\end{equation}
Otherwise, the jammer employs a strategy $\boldsymbol{w}_{1b} = \pm \boldsymbol{b}$ where the sign of $\boldsymbol{w}_{1b}$ matches to $\sign(\lambda_0 - c)$. In such a case, the network adopts a best response strategy
\begin{equation}
	\lambda_1 = c \pm \sqrt{P \cdot \boldsymbol{b}^T \boldsymbol{b} }.
\end{equation}

In summary, if $\lambda_0$ lies between $c - \sqrt{P \cdot \boldsymbol{b}^T \boldsymbol{b}}$ and $c + \sqrt{P \cdot \boldsymbol{b}^T \boldsymbol{b}}$, the repeated game converges to an equilibrium point $(\lambda_0, \boldsymbol{w}_{1a})$ in one iteration. Else, the repeated game converges to an equilibrium point $(\lambda_1, \boldsymbol{w}_{1b})$.
\end{proof}

Since both the network and the jammer converge rationally to an equilibrium presented in Theorem \ref{Thrm: NE-K=1}, there is no incentive for the jammer to employ a pure strategy. This is because the error probability at the FC under such equilibrium solutions is totally independent of the jammer's strategy. In fact, the error probability at the FC in the presence of a jammer is identical to that in the absence of a jammer (i.e., $\boldsymbol{w} = \boldsymbol{0}$). 

Given that pure strategies are not beneficial to the jammer, we now investigate if there is any incentive to employ a mixed strategy at the jammer. For the sake of illustration, we consider an example similar to the model in \cite{Nadendla2011}, where the jammer employs a signal $\boldsymbol{w} \sim \mathcal{N}(\boldsymbol{0}, W)$ and admits an average power constraint\footnote{Note that the support of an average power constraint spans over $\mathbb{R}^{L+M}$, unlike the strict power constraint which has a compact support set $\mathcal{W}$.} $\tr(W) \leq P$. In the following lemma, we demonstrate that a simple Gaussian jammer with an average power constraint has a greater impact than that of a pure-strategy equilibrium stated in Theorem \ref{Thrm: NE-K=1}.

\begin{lemma}
	When the network employs its best response (mixed) strategy to the jammer's mixed strategy, the expected utility (average error probability) due to a Gaussian jammer with an average power constraint is always greater than the error probability under pure-strategy equilibrium.
\end{lemma}

\begin{proof}
Let us define a functional
\begin{equation}
	\begin{array}{lcl}
		\Gamma(x) & = & \displaystyle \pi_0 Q \left( \frac{x}{\sqrt{\sigma^2 + \boldsymbol{b}^T W \boldsymbol{b}}} \right) 
		\\[3ex]
		&& \displaystyle \qquad + \pi_1 \left[ 1 - Q \left( \frac{x - a}{\sqrt{\sigma^2 + \boldsymbol{b}^T W \boldsymbol{b}}} \right) \right].
	\end{array}
\end{equation}

Given a fixed threshold $\lambda$ at the FC, the error probability at the FC turns out to be $\tilde{P}_E(\lambda) = \Gamma(\lambda)$. Note that $\tilde{P}_E(\lambda)$ is a quasiconvex\footnote{Proof is similar to our approach in Lemma \ref{Lemma: Quasiconvex P_E(lambda) when K=1}.} function of $\lambda$. In other words, if the network employs a mixed strategy, the optimal (best response) distribution is given by $p(\lambda) = \delta(\lambda^*)$, where $\lambda^* = \displaystyle c + \frac{1}{a} \boldsymbol{b}^T W \boldsymbol{b} \log \frac{\pi_0}{\pi_1}$ is the optimal threshold that minimizes $\tilde{P}_E(\lambda)$, and $\delta(x)$ is a Dirac delta function centered at $x$. Thus, the expected utility (minimum $\tilde{P}_E(\lambda)$) due to a Gaussian jammer is
\begin{equation}
	U(W) = \Gamma \left( \displaystyle c + \boldsymbol{b}^T W \boldsymbol{b} \frac{1}{a} \log \frac{\pi_0}{\pi_1} \right).
\end{equation}

Note that $U(W)$ is a quasiconvex\footnote{The proof is similar to our approach in Lemma \ref{Lemma: Quasiconcave P_E(w) when K=1}.} function of $W$, with its minimum at $W$ being an all-zero matrix. In other words,
\begin{equation}
	U(W) \geq P_E(\lambda^*, w^*),
\end{equation}
where $P_E(\lambda^*, w^*)$ is given in Equation \ref{Eqn: Pe-FC-equilibrium-K=1}. Consequently, the jammer has every incentive to use a mixed strategy rather than employing a deterministic (pure) strategy. 
\end{proof}

\section{Conclusion and Future Work}
We have modeled the interaction between a centralized detection network and a jammer as a zero-sum game, and found a family of pure strategy Nash equilibria in a closed-form. We have also shown that both the players converge to one of the equilibrium points proposed, in a perfectly-observable repeated game irrespective of the order of their play. Given that pure-strategy jamming attacks have no impact on network performance, we demonstrated that even a simple Gaussian jammer with average power constraints achieves a greater expected utility (average error probability due to mixed strategies) than in the case of pure-strategy equilibria. In the future, we will investigate mixed-strategy equilibria in our proposed framework under strict power constraints. Furthermore, we will consider practical incomplete-information games where both the network and the jammer has partial knowledge about the channel gains. We will also study the effects of receiver diversity at the FC, on the network performance in the presence of a jammer.


\bibliographystyle{IEEEtran}
\bibliography{IEEEabrv,references}

%
%

\end{document}